\documentclass[journal,twoside,web]{ieeecolor}
\usepackage{generic}
\usepackage{cite}
\usepackage{amsmath,amssymb,amsfonts}
\usepackage{algorithmic}
\usepackage[pdftex]{graphicx}
\usepackage{epstopdf}
\usepackage{textcomp}
\usepackage{tikz}
\usepackage{xcolor}
\usepackage{subfigure}
\usetikzlibrary{arrows,shapes,snakes,calc}
\colorlet{TTColor}{blue!80!}
\definecolor{hanblue}{rgb}{0.27, 0.42, 0.81}
\definecolor{frenchblue}{rgb}{0.0, 0.45, 0.73}
\definecolor{midnightblue}{rgb}{0.1, 0.1, 0.44}
\colorlet{SubTitleColor}{midnightblue}

\newtheorem{theorem}{Theorem}
\newtheorem{lemma}{Lemma}
\newtheorem{assumption}{Assumption}
\newtheorem{definition}{Definition}

\def\BibTeX{{\rm B\kern-.05em{\sc i\kern-.025em b}\kern-.08em
    T\kern-.1667em\lower.7ex\hbox{E}\kern-.125emX}}
\begin{document}
% \title{Towards practical distributed average tracking: input delay, package dropouts, and reference noise}
\title{Delay and Packet-Drop Tolerant Multi-Stage Distributed Average Tracking in Mean Square}
\author{Fei Chen, Changjiang Chen, Ge Guo, Changchun Hua, and Guanrong Chen
\thanks{F.~Chen and G.~Guo are with the State Key Laboratory of Synthetical Automation for Process Industries, Northeastern University, Shenyang, 110004, China and School of Control Engineering, Northeastern University at Qinhuangdao, Qinhuangdao, 066004, China. C.~Chen is with the Department of Automation, Xiamen University, Xiamen, 361005, China. C.~Hua is with the Institute of Electrical Engineering, Yanshan University, Qinhuangdao, 066004, China. G.~Chen is with the Department of Electrical Engineering, City University of Hong Kong, Hong Kong, China.} 
\thanks{A preliminary version of this work was presented in Chinese Control Conference \cite{chen2019discrete}.}
}

\maketitle

\begin{abstract}
This paper studies the distributed average tracking problem pertaining to a discrete-time linear time-invariant multi-agent network, which is subject to, concurrently, input delays, random packet-drops, and reference noise. The problem amounts to an integrated design of delay and packet-drop tolerant algorithm and determining the ultimate upper bound of the tracking error between agents' states and the average of the reference signals. The investigation is driven by the goal of devising a practically more attainable average tracking algorithm, thereby extending the existing work in the literature which largely ignored the aforementioned uncertainties. For this purpose,
a blend of techniques from Kalman filtering, multi-stage consensus filtering, and predictive control is employed, which gives rise to a simple yet comepelling distributed average tracking algorithm that is robust to initialization error and allows the trade-off between communication/computation cost and stationary-state tracking error. Due to the inherent coupling among different control components, convergence analysis is significantly challenging. Nevertheless, it is revealed that the allowable values of the algorithm parameters rely upon the maximal degree of an expected network, while the convergence speed depends upon the second smallest eigenvalue of the same network's topology. The effectiveness of the theoretical results is verified by a numerical example.

% This paper studies a distributed average tracking ( DAT ) problem for reference signals subject to both process and measurement noise, which are with expected steady state, under an undirected network with communication delays and random packet dropout. In order to attenuate the effect of the noise, we construct a Kalman filter for each agent. And a state predictor for each agent is applied to compensate the effect of time delays, a expected network is cited to address packet dropout. Based on the above, we design a discrete-time consensus-based DAT algorithm. It is shown that the proposed algorithm can track the average of the reference signals with small steady-state error. The effectiveness of these theoretical results is verified by a numerical example.
\end{abstract}


\section{Introduction}

For a multi-agent plant operating through a network of devices, the capability of distributed average tracking (DAT), measured by the tracking error between each agent's state and the average of a set of reference signals via a distributed protocol, depends on the agent dynamics, the network topology, the class of control algorithms, as well as the reference signals. It has been recognized that when the agent dynamics, the reference signals, and the network topology are given, and the control algorithm has been designed, the reference noise, input delay, and network transmission failures will also lead to degrading control performance. This paper considers the DAT problem for discrete-time multi-agent systems, in which both the agents' states and the control inputs are updated in a discrete-time manner. The effect of linear time-invariant agent dynamics, noise-corrupted reference signals, and unreliable transmission networks subject to random packet-drop are investigated. Each agent's local input will be implemented via a multi-stage algorithm. The objective is to investigate what may affect the tracking error in this
setting, and whether it is possible to achieve practical DAT, i.e., the stationary-state tracking error can be made arbitrarily small, and how.

The capability of distributed average tracking is a significant attribute of multi-agent systems, which has proven useful for distributed sensor fusion \cite{spanos2005distributed,olfati2005consensus,bai2011distributed}, distributed optimization \cite{rahili2017distributed}, and multi-agent coordination \cite{yang2008multi,sun2007swarming,chen2017connection,chen2019minimum,mei2015distributed}. For single-integrator plants, a consensus algorithm and a proportional-integral algorithm are investigated respectively in \cite{spanos2005dynamic} and \cite{freeman2006stability}, wherein both algorithms could track the average of stationary references with zero tracking errors. The proportional-integral control offers additional robustness against initialization errors. Meanwhile, more advanced design methods have been exploited to track time-varying references \cite{bai2010robust}, sinusoid references with unknown frequencies \cite{bai2016two}, and arbitrary references with bounded derivatives \cite{chen2012distributed}. Recently, the study on DAT has been expanded to handle complicated agent dynamics, e.g., double-integrator dynamics \cite{ghapani2017distributed,chen2015distributed}, generic linear dynamics \cite{chen2017connection,zhao2017distributed}, and nonlinear dynamics \cite{chen2014distributed,zhao2018distributed}, with performance analysis \cite{van2015optimal,van2015exploiting,yuan2012decentralised}, privacy requirements \cite{kia2015dynamic}, and for balanced directed networks \cite{li2019dynamic,li2015dynamic}. By introducing a ``damping'' factor, the algorithm of \cite{montijano2014robust} ensures DAT with small errors while being robust against initialization errors. Inspired by the proportional algorithm, a multi-stage DAT algorithm was lately proposed in \cite{franceschelli2019multi} based upon a cascade of DAT filters, which is capable of achieving DAT with bounded errors. For more details on DAT, a recent tutorial is available \cite{kia2019tutorial}.

In spite of significant progress on DAT, the study on practical issues, such as delay and noise, is only to emerge \cite{moradian2018robustness}. Indeed, it is generally recognized, and intuitively clear, that the convergence of DAT algorithms can be constrained by transmission failures, input delays, as well as reference noise, which all likely result in negative effect on the convergence of the closed-loop system. For linear systems with small input delays, the control algorithm without delay compensation might still work, since linear systems possess a certain robustness margin to small delays \cite{niculescu2002delay}; yet the convergence will generically fail for relatively large delays. In practice, a reference signal might represent a target or a dynamic process, for which the measurement is inevitably corrupted by random noise. However, the effect of reference noise has been commonly ignored. Apart from that, since the communication network is shared by all agents, packet-drops ubiquitously exist, and therefore should be fully addressed, particularly when the data transmission rate is large. 
%Kalman filtering has proved to be effective in estimation of noisy signal which is in terms of the tracking .

% Another fundamental issue in DAT is caused by the reference noises, which
% to estimate and track the state of targets (or dynamic processes) of interest that evolve in the sensing field.

% For linear systems with state and/or input delays that are relatively small compared with the plant's time scales, finite- dimensional feedback laws designed by standard algorithms such as pole assignment and linear quadratic regulator are still appli- cable since every linear system possesses a robustness margin to small delays (Niculescu, 2001).

% Some of these above aforementioned works address important practical considerations such as the discrete-time distributed average tracking over time-varying network and nonlinear distributed average tracking. However, all aforementioned works are considered in ideal communication network. In practice, all signals are subject to noise, either deterministic or random, and communication networks may be with communication delays and random packet dropout between communication. It is well known that noise, time delays and packet dropout have significant effect on the stability of control systems.  considered the continuous-time and discrete-time distributed average tracking in the presence of the fixed communication time delays. There are almost none works taking noise, time delays, packet dropout all into consider. This paper intends to address DAT problem with these inevitable disturb in multi-agents system ( MAS ).

This paper proposes a practical DAT design which can concurrently tolerate input delays, random packet-drops, and reference noise. For this purpose, a blend of the techniques from multi-stage consensus filtering, predictive control, and Kalman filtering is employed. This work extends the existing work to a more realistic setting where the idealized assumptions, which are seldom possible in practice, are removed. To the best of the authors' knowledge, this work presents the first multi-stage design that takes reference noise, input delays, packet-drops all into consideration, which are perceived as main sources of control design difficulty for multi-agent systems. With this defining feature, the analysis reveals that an expected network topology plays a central role in ensuring the convergence of the proposed DAT algorithm, wherein the allowable values of the algorithm parameters rely upon the maximal degree of the expected network, while the convergence speed depends on the second smallest eigenvalue of the same network's topology. It should be noted that due to the additional algorithm components and their inherent couplings, the convergence analysis is significantly more challenging than the existing ones.

% 17.	A fact perceived as a main source of control design difficulty for multi-variable systems. 

% which is more readily usable than existing DAT algorithms.

% In this paper, we study the DAT algorithm \cite{franceschelli2019multi} in the presence of communication delays, random packet dropout, and the reference signals are subject to both process and measurement noise. The contribution of this paper is summarized as follows.
% \begin{itemize}
% \item [1.] We apply a state predictor for each agent to compensate the effect of time delays. This predictor allows each agent to predict the delayed states between two consecutive update times. The predicted states are then employed in the distributed average tracking algorithm to guarantee the desired convergence properties.
% \item [2.] In order to attenuate the effect of the noise, we construct a Kalman filter for each agent. The filter allows each agent to estimate the reference signals subject to both process and measurement noise, then the estimates are applied to DAT algorithm.
% \item [3.] A expected network is employed to address random packet dropout. Then we can guarantee the expected convergence properties.
% \end{itemize}

The rest of this paper is organized as follows. In Section~\ref{sec:prel}, notation and mathematical preliminaries are presented. In Section~\ref{sec:problDescr}, the problem is defined. In Section~\ref{sec:algDes}, the DAT algorithm is designed with the aid of Kalman filtering, predictive control, and multi-stage consensus filtering. The performance of the proposed algorithm is analyzed in Section~\ref{sec:convAnal}. Section~\ref{sec:sim} presents numerical examples to verify the theoretical results. Finally, Section~\ref{sec:concl} concludes the paper.
% In Section 4, we propose a control algorithm with the aid of Kalman filter, which includes the control scheme compensating the effect of time delays. A expected network is proposed to address the random packet dropout. We investigate the performance of the proposed DAT algorithm in Section 5. Section 6 presents a numerical examples which illustrate the validity of the theoretical results. Finally, Section 7 concludes this paper.

\section{Preliminaries}
\label{sec:prel}
\subsection{Notation}
%The notation used throughout this paper is fairly standard. 
Let $\mathbb{R}^+$ denote the set of real numbers and $\mathbb Z^+$ the set of positive integers.  Let $\mathbb R^n$ and $\mathbb R^{m \times n}$ denote respectively the set of $n$-dimensional real vectors and the set of $m \times n$ real matrices. Let $\boldsymbol {\mathrm I}_n \in \mathbb R^{n \times n}$ be the $n$-dimensional identity  matrix, $\boldsymbol 1_n \in \mathbb R^n$ the $n$-dimensional vector with all ones, and $\boldsymbol 1_{m \times n} \in \mathbb R^{m \times n}$ the $m \times n$ matrix with all ones.
%The subscripts of $I_n$, $\boldsymbol 1_n$ and $\boldsymbol 1_{m \times n}$ might be dropped if there is no confusion in the context.
For a vector $\boldsymbol x \in \mathbb R^{n}$, the norm used is defined as $\| \boldsymbol x \|_2 \triangleq (|x_1|^2+ \dots +|x_n|^2)^{1/2}$. The transpose of matrix $A$ is denoted by $A^{\mathrm T}$. The diagonal matrix with $a_i$ $(i=1, 2, \dots, n)$ being
its $i$th diagonal element is denoted by $\text{diag} \{ a_1, a_2,\dots, a_n \}$. Let $A^{-1}$ denote the inverse of matrix $A$. The smallest and largest eigenvalues of $A$ are given respectively by $\lambda_{ \min }(A)$ and $\lambda_{ \max }(A)$. Let $\| A \|_2 \triangleq \sqrt{\lambda_{ \max }(A^{\mathrm T}A)}$. It is assumed that all the vectors and matrices have compatible dimensions, which may not be shown if clear from the context.
For a set $S$, let $|S|$ denote its cardinality, i.e., the number of elements in $S$. Let $\mathbb E(\cdot)$ be the mathematical expectation and $\mathbb P(\cdot)$ be the probability function. The normal probability distribution is denoted by $N(\cdot)$.

% In this paper, $\mathbb R$ denotes the set of all real numbers, $\mathbb R^+$ denotes the set of all positive real numbers, and $\mathbb Z^+$ denotes the set of all positive integers. Let $\mathbb R^n$ and $\mathbb R^{m \times n}$ denote the set of $n$-dimensional real vectors and the set of $m \times n$ real matrices, respectively. Let $I_n \in \mathbb R^{n \times n}$ be the $n$-dimentional identity  matrix, $\boldsymbol 1_n \in \mathbb R^n$ be the $n$-dimensional vector with all ones and $\boldsymbol 1_{m \times n} \in \mathbb R^{m \times n}$ be the $m \times n$-dimensional vector with all ones. The subscripts of $I_n$, $\boldsymbol 1_n$ and $\boldsymbol 1_{m \times n}$ might be omitted if there is no confusion in the context.

%  For a matrix $A \in \mathbb R^{m \times n}$, $A^\mathrm T$ denotes its transpose. If $m=n$, let $\text {Trace} (A) \in \mathbb R$ be the trace of $A$, i.e., $\text {Trace} (A)=\sum_{i=1}^{n}a_{ii}$, where $a_{ij}$ is the $(i,j)$th element of $A$. Let $A^{-1}$ denote the inverse of $A$. The smallest and largest eigenvalues of a matrix $A$ are denoted by $\lambda_{ \min }(A)$ and $\lambda_{ \max }(A)$, respectively. And $\| A \|_2 \triangleq \sqrt{\lambda_{ \max }(A^{\mathrm T}A)}$. For a set $S$, $|S|$ denotes its cardinality. Let $\mathbb E(\cdot)$ denote the mathematical expectation and $\mathbb P(\cdot)$ the occurrence probability. The normal probability distribution is denoted by $N(\cdot)$.
\subsection{Graph Theory}
A graph is defined by $\mathcal G \triangleq (\mathcal V,\mathcal E)$, where $\mathcal V$ is the set of nodes and $\mathcal E \subseteq {\mathcal V \times \mathcal V}$ is the set of edges. A graph is undirected if $(i,j) \in \mathcal E \Longleftrightarrow (j,i) \in \mathcal E$ for all $i,j \in \mathcal V$. This paper considers undirected graphs. For node $i$, the set of its neighbors is defined as $\mathcal{N}_i = \{j\in \mathcal {V} \;|\; (j,i) \in \mathcal {E}\}$. The adjacency matrix of $\mathcal G$ is given by $ A=[a_{ij}] \in \mathbb R^{N \times N}$, where $a_{ij} = 1$ if $(j,i) \in \mathcal E$, and $a_{ij}=0$ otherwise. The degree of node $i$ is defined as $ d_i=  \sum_{j=1}^N a_{ij} =| \mathcal N_i|$. The maximum degree  of $\mathcal G$ is given by $d_{\max}=\max \{d_1, d_2, \dots, d_N \}$. The degree matrix is given by $D = \mathrm{diag}\{d_1, d_2, \dots, d_N\}\in \mathbb{R}^{N\times N}$. The Laplacian matrix of $\mathcal{G}$ is defined as $L = D-A \in \mathbb{R}^{N\times N}$. A graph is connected if, for any pair $\{i,j\}$, there exists a path connecting $i$ to $j$. For graphs $\mathcal G =(\mathcal V, \mathcal E)$ and $\mathcal G' =(\mathcal V, \mathcal E')$ with the same node set, their union is given by $\mathcal G \cup \mathcal G' \triangleq (\mathcal{V}, \mathcal E \cup \mathcal E')$.

\subsection{Observability and Controllability}
% In this section, we introduce some concepts from control theory. All matrices in the following are assumed to be of appropriate dimensions.

\begin{definition}[\hspace*{-4pt}\cite{karvonen2014stability}]\label{def1}
A matrix pair $[F(k), G(k)]$ with $k \in \mathbb{Z}^+$ is said to be completely observable if the observability Gramian
\begin{equation*}
O(k,l) := \sum_{i=l}^k \left( \prod_{j=l}^{i-1} F(j) \right)^{\mathrm T}G^\mathrm{T}(i)G(i) \left( \prod_{j=l}^{i-1}F(j) \right),
\end{equation*}
defined for $l < k $, is positive definite for some $k$ and $l$. Furthermore, the pair is said to be uniformly (completely) observable if there exists a positive integer $n$ and positive constants $\alpha_1$ and $\alpha_2$ such that
\begin{equation*}
0 \le \alpha_1 \boldsymbol {\mathrm I} \le O(k, k-n) \le \alpha_2 \boldsymbol {\mathrm I},
\end{equation*}
for all $k \ge n$.
\end{definition}
\begin{definition}[\hspace*{-4pt}\cite{karvonen2014stability}]\label{def2}
A matrix pair $[F(k), G(k)]$ with $k \in \mathbb{Z}^+$ is said to be completely controllable if the controllability Gramian
\begin{equation*}
C(k, l):= \sum_{i=l}^{k-1} \left( \prod_{j=i+1}^{k-1}F(j) \right) G(i)G^{\mathrm T}(i) \left( \prod_{j=i+1}^{k-1}F(j) \right)^\mathrm{T},
\end{equation*}
defined for $l < k$, is positive definite for some $k$ and $l$. Furthermore, the pair is said to be uniformly (completely) controllable if there exists a positive integer $n$ and positive constants $\beta_1$ and $\beta_2$ such that
\begin{equation*}
0 \le \beta_1 \boldsymbol {\mathrm I} \le C(k, k-n) \le \beta_2 \boldsymbol {\mathrm I},
\end{equation*}
for all $k \ge n$.
\end{definition}
% Consider the following discrete time-varying system:
% \begin{equation}\label{eq2-5}
% x(K+1)=A(K)x(K)+B(K)u(K),
% \end{equation}
% where $x(K) \in \mathbb R^n$, $A(K)$ and $B(K)$ are two matrices of appropriate sizes. The state transition matrix of \eqref{eq2-5} from time $p$ to $q$, $q > p$, is
% \begin{equation}\label{eq2-6}
% \Phi(q,p) := \prod_{K=p}^{q-1}A(K)
% \end{equation}
% \begin{definition}[\cite{karvonen2014stability}]\label{def3}
% The homogeneous part of \eqref{eq2-5} is exponentially stable, if there are $c_1, c_2 \in \mathbb R^+$, which are independent of $K$, such that  $\| \Phi(K+1, 0) \|_2 \le c_1e^{-c_2K}$, for all $K \ge 0$.
% \end{definition}
% For the following discrete time-invariant system:
% \begin{equation}\label{eq2-7}
% x(K+1)=Ax(K),
% \end{equation}
% where $x(K) \in \mathbb R^n$, $A$ is a matrix of appropriate size. The asymptotically stability is defined as follows.
% \begin{definition}[\cite{chen1998linear}]\label{def4}
% The system \eqref{eq2-7} is asymptotically stable if every finite initial state excites a bounded response, which, in addition, approaches $0$ as $K \to \infty$.
% \end{definition}
% \begin{lemma}[\cite{chen1998linear}]\label{lem1}
% The system \eqref{eq2-7} is asymptotically stable if and only if all eigenvalues of $A$ have magnitudes less than $1$.
% \end{lemma}

\section{Problem Description}
\label{sec:problDescr}

Consider a network of $N \in \mathbb Z^+$ agents, labeled from $1$ to $N$. Agent $i$, $i=1,\dots,N$, follows the following discrete-time multi-stage dynamics:
\begin{align}\label{eq3-1}
x_i^1(k+1)&=x_i^1(k)+u_i^1(k-\tau) \nonumber \\
x_i^2(k+1)&=x_i^2(k)+u_i^2(k-\tau) \nonumber \\
& \; \; \vdots \nonumber \\
x_i^n(k+1)&=x_i^n(k)+u_i^n(k-\tau),
\end{align}
where $x_i^p(k)$ is the state of agent $i$ at stage $p$, $u_i^p(k-\tau)$ is the input of agent $i$ at stage $p$ subject to an input delay $\tau$, and $k \in \mathbb{Z}^+$ is the time variable. 
% It will soon become clear that the multi-stage dynamics allows the trade-off between the communication/computation cost and the tracking performance of the proposed algorithm.  
A graph $\mathcal G \triangleq (\mathcal{V}, \mathcal{E} )$ is used to describe the information flows among the agents, where $\mathcal{V} \triangleq \{ 1, \dots, N \}$ is the node set and $\mathcal E \triangleq \{ (i,j) \, | \, \text{node $i$ can share information with $j$}, \, i,j = 1, \dots, N \}$ is the edge set. Throughout this paper, it is assumed that graph $\mathcal{G}$ is connected. %, which  is formally stated in the following assumption.
% \begin{assumption}\label{ass1}
% Graph $\mathcal G$ is connected.
% \end{assumption}
% Assumption~\ref{ass1} is a common assumption on most literature on multi-agent cooperative control under an undirected network topology.

Because information is exchanged through the network, there usually exist packet-drops particularly when the data transmission rate is high. It is common and reasonable to delineate packet-drops by independent Bernoulli processes. Specifically, let $\theta_{ij}$ be a random variable indicating whether the transmission between two neighboring agents $i$ and $j$ is successful, i.e.,
\begin{align}\label{eq3-2}
	\theta_{ij}=
		\begin{cases}
			0, & \text{with probability $p_{ij}$,}\\
			1, & \text{with probability $1-p_{ij}$,}
		\end{cases}
\end{align}
where $0 \leq p_{ij}< 1$ is the packet-drop rate. For non-neighboring agents, $\theta_{ij}=0$ with probability~$1$. Due to the random packet-drops, the {\it de facto} information exchange network, denoted by $\widetilde{\mathcal{G}} \triangleq (\mathcal{V},\widetilde{\mathcal{E}})$, is by nature a random network. At each time instant, $\widetilde{\mathcal{G}}$ takes a value from the set $\{\mathcal{G}_1,\dots, \mathcal{G}_s\}$, $s=2^{|\mathcal{E}|}$, 
% $s=2^{\sum_{i=1}^N \frac{|\mathcal{N}_i|}{2}}$, 
with the probability
\begin{equation}\label{eq3-3}
\mathbb P(\widetilde{\mathcal G} = \mathcal G_m)= \prod_{i=1,i<j}([\theta_{ij}]_m(1-p_{ij})+(1-[\theta_{ij}]_m)p_{ij}),
\end{equation}
where $m=1,\dots,s$ and  and $[\cdot]_m$ means that it takes values from the graph indexed by $m$.

Each agent has a reference signal $r_i$, which is governed by
\begin{equation}\label{eq3-4}
\begin{split}
&r_i(k+1)=r_i(k)+v_i(k)+w_i(k)\\
&z_i(k+1)=h_ir_i(k+1)+\vartheta_i(k+1),
\end{split}
\end{equation}
where $v_i(k) \in \mathbb R$ is the input, $z_i(k) \in \mathbb R$ is the measurement, and $h_i \in \mathbb R^+$ is the measurement gain. The process noise $w_i(k) \in \mathbb R$ and measurement noise $\vartheta_i(k) \in \mathbb R$ follow independent normal probability distributions, i.e.,
\begin{equation*}
\begin{split}
&w_i(k) \sim N(0,\phi_i(k))\\
&\vartheta_i(k) \sim N(0,\psi_i(k)),
\end{split}
\end{equation*}
where $\phi_i(k) \triangleq \mathbb E[w_i(k)^2]$ is the variance of the process noise, and $\psi_i(k) \triangleq \mathbb E[\vartheta_i(k)^2]$ is the variance of the measurement noise. The following assumption is made on the reference signals.
\begin{assumption}\label{ass2}
The reference signals satisfy the following properties:
\begin{enumerate}
	\item[(i)] the expectation $\mathbb E[r_i(k)]$ approaches a constant, as $k \to \infty$;
	\item[(ii)] $\rho_1 \le \phi_i(k) \le \rho_2 $ and $ ~\mu_1 \le \psi_i(k) \le \mu_2$,
where $\rho_1, \rho_2, \mu_1, \mu_2 \in \mathbb R^+$.
\end{enumerate}
\end{assumption}
% \begin{assumption}\label{ass3}

% \end{assumption}

The primary objective of this paper is to design distributed input sequence for the system \eqref{eq3-1} such that all agents can track the average of the $N$ noisy reference inputs in the sense that, for all $i=1, 2, \dots, N$,
\begin{equation}\label{eq3-5}
\limsup_{k \to \infty} \left \|\mathbb E \left[ x_i^n(k)-\frac {1}{N} \sum_{i=1}^Nr_i(k) \right] \right \|_2 \le \delta,
\end{equation}
where $\delta$ is a pre-desired constant, which can be arbitrarily small.

\section{Algorithm Design}
\label{sec:algDes}

In this section, the control input sequence is designed for the multi-stage system \eqref{eq3-1} to track the average signal of the noisy references \eqref{eq3-4}, which gives the following DAT algorithm: 
\begin{align}\label{eq4-1}
	u_i^1(k)&=-\epsilon \sum_{j=1}^N \theta_{ij} a_{ij} \left (\hat x_i^1(k|k-\tau)-\hat x_j^1(k|k-\tau) \right ) \nonumber \\
           &\quad +\alpha \left(\hat r_i(k|k-\tau)-\hat x_i^1(k|k-\tau) \right ) \nonumber \\
    u_i^2(k) &=-\epsilon \sum_{j=1}^N \theta_{ij} a_{ij} \left (\hat x_i^2(k|k-\tau)-\hat x_j^2(k|k-\tau) \right ) \nonumber \\
 		  & \quad +\alpha \left(\hat x_i^1(k|k-\tau)-\hat x_i^2(k|k-\tau) \right ) \nonumber \\
 		  &\hspace*{.2\textwidth} \vdots \nonumber \\
    u_i^n(k) &=-\epsilon \sum_{j=1}^N \theta_{ij} a_{ij} \left (\hat x_i^n(k|k-\tau)- \hat x_j^n(k|k-\tau) \right ) \nonumber \\
		  &\quad +\alpha \left(\hat x_i^{n-1}(k|k-\tau)-\hat x_i^n(k|k-\tau) \right),
\end{align}
where $\hat x_i^p(k|k-\tau)$ and $\hat r_i(k|k-\tau)$ are respectively the predicted states of agent $i$ and reference $i$ at time instant $k$ using the measurement information up to time instant $k-\tau$, and $\epsilon >0$ and $\alpha >0$ are two gain parameters to be designed.

The agent state predictor is given by
\begin{equation}\label{eq4-2}
\begin{split}
 &\hat x_i^1(k-\tau+1|k-\tau)\\
=\, &\hat x_i^1(k-\tau|k-\tau-1)\\
 &+k_x \left (x_i^1(k-\tau)-\hat x_i^1(k-\tau|k-\tau-1) \right)\\
 &+u_i^1(k-\tau)\\
 &\hat x_i^2(k-\tau+1|k-\tau)\\
=\, &\hat x_i^2(k-\tau|k-\tau-1)\\
 &+k_x \left (x_i^2(k-\tau)-\hat x_i^2(k-\tau|k-\tau-1) \right)\\
 &+u_i^2(k-\tau)\\
 & \hspace*{.2\textwidth} \vdots\\
 &\hat x_i^n(k-\tau+1|k-\tau)\\
=\, &\hat x_i^n(k-\tau|k-\tau-1)\\
 &+k_x \left (x_i^n(k-\tau)-\hat x_i^n(k-\tau|k-\tau-1) \right)\\
 &+u_i^n(k-\tau),\\
\end{split}
\end{equation}
where $k_x>0$ is the predictor gain, while the states of agent~$i$ from $k-\tau+2$ to $k$ are predicted by
\begin{equation}\label{eq4-3}
\begin{split}
 &\hat x_i^1(k-\tau+l|k-\tau)\\
=\,&\hat x_i^1(k-\tau+l-1|k-\tau)+u_i^1(k-\tau+l-1)\\
 &\hat x_i^2(k-\tau+l|k-\tau)\\
=\,&\hat x_i^2(k-\tau+l-1|k-\tau)+u_i^2(k-\tau+l-1)\\
 & \hspace*{.2\textwidth} \vdots\\
 &\hat x_i^n(k-\tau+l|k-\tau)\\
=\,&\hat x_i^n(k-\tau+l-1|k-\tau)+u_i^n(k-\tau+l-1),\\
 &~~~~~~~~~~~~~~~~~~~~~~~~~~~~~~~~~~~~~~~~~~~~l=2,\dots,\tau.
\end{split}
\end{equation}

Because the reference signals are subject to both input delay and noise, its design needs a combination of the predictive control and Kalman filtering techniques. Specifically, the following predictor for the reference signals is designed:
\begin{equation*}
\begin{split}
 &\hat r_i(k-\tau+1|k-\tau)\\
=\,&\hat r_i(k-\tau|k-\tau-1)\\
 &+k_r \left (\hat r_i(k-\tau)-\hat r_i(k-\tau|k-\tau-1) \right)\\
 &+v_i(k-\tau),\\
 &\hat r_i(k-\tau+l|k-\tau)\\
=\,&\hat r_i(k-\tau+l-1|k-\tau)+v_i(k-\tau+l-1),\\
 &~~~~~~~~~~~~~~~~~~~~~~~~~~~~~~~~~~~~~~~~~~l=2,\dots,\tau,
\end{split}
\end{equation*}
where $k_r>0$ is a predictor gain and $\hat r_i(k)$ is the estimate obtained via the Kalman filter
\begin{subequations}\label{eq4-4}
\begin{alignat}{1}
 \hat r_i^-(k+1)&=\hat r_i(k)+v_i(k)\label{eq4-4:a}\\
      p_i^-(k+1)&=p_i(k)+\phi_i(k)\label{eq4-4:b}\\
        k_i(k+1)&=\frac{h_ip_i^-(k+1)}{h_i^2p_i^-(k+1)+\psi_i(k+1)}\label{eq4-4:c}\\
   \hat r_i(k+1)&=\hat r_i^-(k+1)+k_i(k+1)(z_i(k+1)-h_i \hat r_i^-(k+1))\label{eq4-4:d}\\
        p_i(k+1)&=(1-k_i(k+1)h_i)p_i^-(k+1)\label{eq4-4:e},
\end{alignat}
\end{subequations}
%\begin{equation}\label{eq4-5}
%\hat r_i^-(K+1)=\hat r_i(K)+v_i(K)
%\end{equation}
%\begin{equation}\label{eq4-6}
%p_i^-(K+1)=p_i(K)+Q_i(K)
%\end{equation}
%\begin{equation}\label{eq4-7}
%K_i(K+1)=\frac{h_ip_i^-(K+1)}{h_i^2p_i^-(K+1)+R_i(K+1)}
%\end{equation}
%\begin{equation}\label{eq4-8}
%\hat r_i(K+1)=\hat r_i^-(K+1)+K_i(K+1)(z_i(K+1)-h_i \hat r_i^-(K+1))
%\end{equation}
%\begin{equation}\label{eq4-9}
%p_i(K+1)=(1-K_i(K+1)h_i)p_i^-(K+1),
%\end{equation}
in which $k_i(k) \in \mathbb R$ is the Kalman gain, $\hat r_i^-(k) \in \mathbb R$ and $\hat r_i(k) \in \mathbb R$ are, respectively, the priori and posteriori estimates of $r_i(k)$, $p_i^-(k) \triangleq \mathbb E \left [(r_i(k)-\hat r_i^-(k))^2 \right ]$ is the mean-squared priori estimate error, and $p_i(k) \triangleq \mathbb E \left [(r_i(k)-\hat r_i(k))^2 \right ]$ is the mean-squared posteriori estimate error. The initial states of the Kalman filter, $\hat r_i(0)$ and $p_i(0) \geq 0$, are chosen randomly.

\begin{figure}[htp!]
	\centering
	\hspace*{-40pt}\includegraphics[width=.4\textwidth]{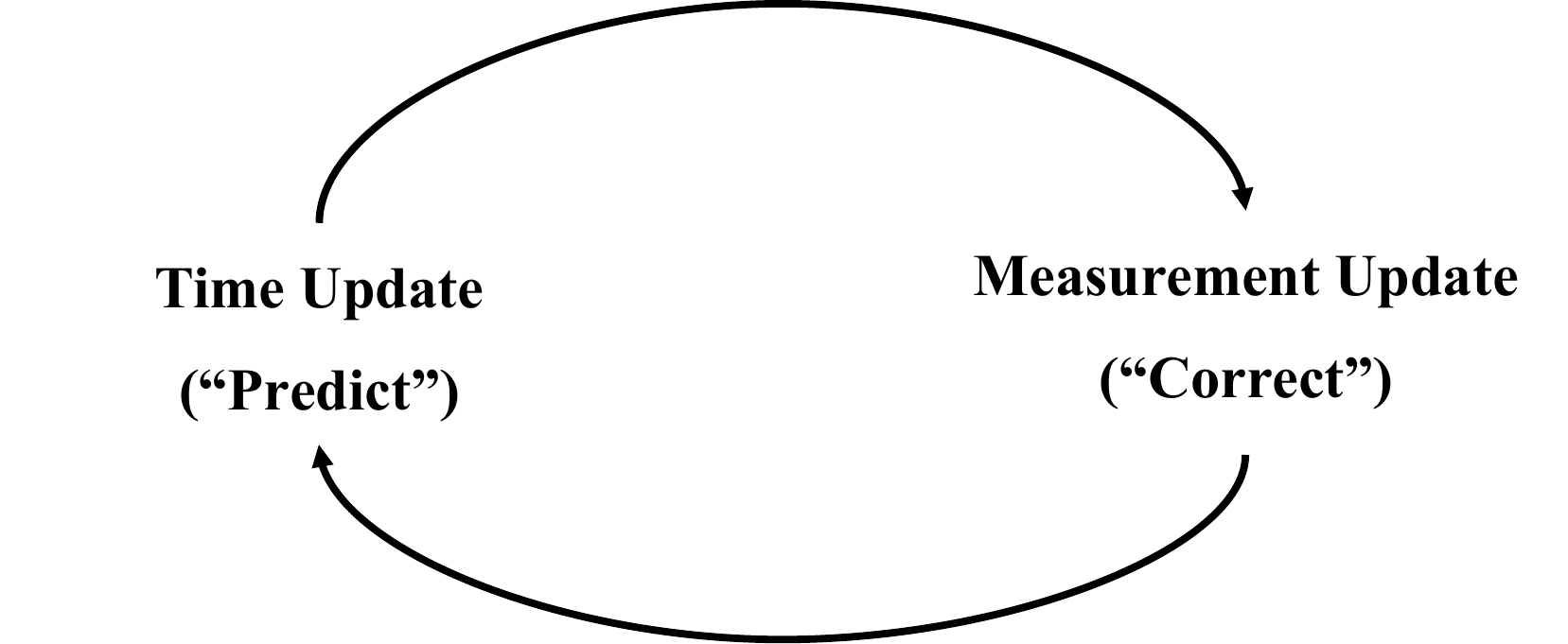}
	\caption{The estimate process of the discrete-time Kalman filter.}\label{fig:4-1}
\end{figure}
The estimate process alluded to above is delineated in Fig.~\ref{fig:4-1}. As shown, the Kalman filter performs two operations: time update and measurement update. The time update equations are \eqref{eq4-4:a} and \eqref{eq4-4:b}, while the measurement update equations are \eqref{eq4-4:c}--\eqref{eq4-4:e}. The time update equations ``predict" the state and estimate errors at time $k+1$ from those at time $k$, while the measurement update equations adjust the priori estimate by using the measurement $z_i(k)$ according to the Kalman gain $k_i(k)$, which is obtained by minimizing $p_i(k)$.

Substituting the input \eqref{eq4-1} into system \eqref{eq3-1} leads to the closed-loop system
\begin{equation}\label{eq4-5}
\begin{split}
 &x_i^1(k+1)\\
%=&x_i^1(K)+u_i^1(K)\\
=\, &x_i^1(k)-\epsilon \sum_{j=1}^N \theta_{ij} a_{ij} \left (\hat x_i^1(k|k-\tau)-\hat x_j^1(k|k-\tau) \right )\\
           &+\alpha \left(\hat r_i(k|k-\tau)-\hat x_i^1(k|k-\tau) \right )\\
 &x_i^2(k+1)\\
%=&x_i^2(K)+u_i^2(K)\\
=\, &x_i^2(k)-\epsilon \sum_{j=1}^N \theta_{ij} a_{ij} \left (\hat x_i^2(k|k-\tau)-\hat x_j^2(k|k-\tau) \right )\\
 &+\alpha \left(\hat x_i^1(k|k-\tau)-\hat x_i^2(k|k-\tau) \right )\\
 &\hspace*{.2\textwidth}\vdots\\
 &x_i^n(k+1)\\
%=&x_i^n(K)+u_i^n(K)\\
=\,&x_i^n(k)-\epsilon \sum_{j=1}^N \theta_{ij} a_{ij} \left (\hat x_i^n(k|k-\tau)- \hat x_j^n(k|k-\tau) \right )\\
&+\alpha \left(\hat x_i^{n-1}(k|k-\tau)-\hat x_i^n(k|k-\tau) \right ),
\end{split}
\end{equation}
which can be written in a vector format as
\begin{equation}\label{eq4-6}
\begin{split}
\boldsymbol x^1(k+1)=\,&\boldsymbol x^1(k)-\epsilon \widetilde L \hat {\boldsymbol x}^1(k|k-\tau)\\
         &+\alpha \left (\hat {\boldsymbol r}(k|k-\tau)-\hat {\boldsymbol x}^1(k|k-\tau) \right )\\
\boldsymbol x^2(k+1)=\,&\boldsymbol x^2(k)-\epsilon \widetilde L \hat {\boldsymbol x}^2(k|k-\tau)\\
         &+\alpha \left (\hat {\boldsymbol x}^1(k|k-\tau)-\hat {\boldsymbol x}^2(k|k-\tau) \right )\\
         &\vdots \\
\boldsymbol x^n(k+1)=\,&\boldsymbol x^n(k)-\epsilon \widetilde L \hat {\boldsymbol x}^n(k|k-\tau)\\
         &+\alpha \left (\hat {\boldsymbol x}^{n-1}(k|k-\tau)-\hat {\boldsymbol x}^n(k|k-\tau) \right ),
\end{split}
\end{equation}
where
\begin{equation*}
\begin{split}
              \boldsymbol x^p(k)=\,&[{x_1^p(k)},{x_2^p(k)}, \dots, {x_N^p(k)}]^\mathrm T\\
\hat {\boldsymbol x}^p(k|k-\tau)=\,&[{\hat x_1^p(k|k-\tau)},{\hat x_2^p(k|k-\tau)},\dots, {\hat x_N^p(k|k-\tau)}]^\mathrm T\\%~s=1, \dots, n,\\
  \hat {\boldsymbol r}(k|k-\tau)=\,&[{\hat r_1(k|k-\tau)}, {\hat r_2(k|k-\tau)},\dots, {\hat r_N(k|k-\tau)}]^\mathrm T.
\end{split}
\end{equation*}
Here, $\widetilde L$ is a stochastic Laplacian matrix, changing within the possible set $\{ L_1, L_2, \dots, L_s \}$, where $L_m=\left [l_{ij} \right ]_m$ is the Laplacian matrix associated with $\mathcal G_m$, i.e.,
\begin{equation*}
[\l_{ij}]_m=\left \{
                    \begin{array}{lr}
                    \sum\limits_{q=1}^N [\theta_{iq}]_m a_{iq}, &i=j, \\
                    -[\theta_{ij}]_m a_{ij}, &i \ne j.
                    \end{array}
            \right.
\end{equation*}
For future use, let
\begin{equation*}
\begin{split}
[\overline l_{ij}]_m=\mathbb E[[l_{ij}]_m]=\left \{
                    \begin{array}{lr}
                    \sum\limits_{q=1}^N (1-p_{iq}) a_{iq}, &i=j, \\
                    -(1-p_{ij}) a_{ij}, &i \ne j.
                    \end{array}
                 \right.
\end{split}
\end{equation*}
The following gives a useful result regarding the expecation of the Laplacian matrix $\tilde{L}$. The result will be employed in the convergence analysis in the next section.
\begin{lemma}[\hspace*{-4pt}\cite{zhou2013dynamic}]\label{lem3}
For multi-agent system \eqref{eq4-5} with a connected communication network, the expected Laplacian matrix, $\overline L \triangleq \mathbb E[\widetilde L]$, has only one zero eigenvalue, i.e.,
\begin{equation*}
0=\lambda_{\overline L, 1} < \lambda_{\overline L, 2} \le \dots \le \lambda_{\overline L, N}.
\end{equation*}
\end{lemma}
% From \eqref{eq4-11}, we can know that the Laplacian matrix $\widetilde L$ is a symmetric matrix with zero row sum.

The advantages of the multi-stage DAT system \eqref{eq4-5} are as follows: firstly, it does not involve integral control actions or input derivatives, thus exhibits robustness to initialization errors; secondly, the proposed multi-state scheme enables the possibility of making a trade-off among the communication/computation cost (i.e., the number of stages), the tracking error and the convergence time; thirdly, the proposed scheme takes input delay, packet-drops, and reference noise into consideration, making it more feasible for practical applications.

\section{Convergence Analysis}
\label{sec:convAnal}
In this section, the convergence of the proposed algorithm embedded in system \eqref{eq4-5} is analyzed. It is first to show that the Kalman filter \eqref{eq4-4} is stable.

\begin{lemma}\label{lem2}
If the second part of Assumption \ref{ass2} holds, then the Kalman filter \eqref{eq4-4} is stable, i.e.,
\begin{equation}\label{eq5-1}
\lim_{k \to \infty} \mathbb E[\hat r_i(k)- r_i(k)]=0.
\end{equation}
% \vspace*{1pt}
\end{lemma}
\vspace*{5pt}
\begin{proof}
	It follows from (ii) of Assumption~\ref{ass2} that the reference \eqref{eq3-4} is uniformly observable and uniformly controllable. That is, the matrix pair $[1, \psi_i^{-\frac{1}{2}}(k)h_i]$ is uniformly observable, and the matrix pair $[1, \phi_i^{\frac{1}{2}}(k)]$ is uniformly controllable. The stability result \eqref{eq5-1} then follows immediately from \cite{karvonen2014stability}.
\end{proof}

% In this section, we analyze the performance of the algorithm proposed in last section. We make the following assumption in order to guarantee the stability of the proposed Kalman filter.
% \begin{assumption}\label{ass3}
% $p_i(0) \ge 0,~\rho_1 \le Q_i(K) \le \rho_2 , ~\mu_1 \le R_i(K) \le \mu_2$,
% where $\rho_1, \rho_2, \mu_1, \mu_2 \in \mathbb R^+$.
% \end{assumption}
% \begin{remark}
% The system \eqref{eq3-3} is uniformly observable and uniformly controllable means that the matrix pair $[1, R_i^{-\frac{1}{2}}(K)h_i]$ is uniformly observable, and the matrix pair $[1, Q_i^{\frac{1}{2}}(K)]$ is uniformly controllable. $\rho_1 \le Q_i(K) \le \rho_2 , ~\mu_1 \le R_i(K) \le \mu_2$ guarantees the system \eqref{eq3-3} is uniformly observable and uniformly controllable, $p_i(0) \ge 0$ guarantees $p_i(K) \ge 0$. Assumption \ref{ass3} guarantees the stability of the Kalman filter.
% \end{remark}
% \begin{lemma}[\cite{karvonen2014stability}]\label{lem2}
% If Assumption \ref{ass3} holds for system \eqref{eq3-3}, then the Kalman filter run on agent $i$ is exponentially stable, i.e.,
% \begin{equation}\label{eq5-1}
% \lim_{K \to \infty} \mathbb E[\hat r_i(K)] \to \mathbb E[r_i(K)].
% \end{equation}
% \end{lemma}

In what follows, the performance of the multi-stage DAT system \eqref{eq4-5} is analyzed. Let
\begin{equation*}
\boldsymbol r^* \triangleq \lim_{k \to \infty} \mathbb E[\boldsymbol r(k)]=[r_1^*, r_2^*, \dots, r_N^*]^\mathrm T, ~\boldsymbol r^* \in \mathbb R^N,
\end{equation*}
where $r_i^* \triangleq \lim_{k \to \infty} \mathbb E[r_i(k)], i={1, \dots , N}$. Note that $\boldsymbol{r}^*$ is well defined due to Assumption~\ref{ass2}. The following result characterizes agents' stationary states in terms of $\boldsymbol{r}^*$.

\begin{lemma}\label{lem4}
For the multi-agent system~\eqref{eq4-5} with a connected communication network, if Assumption~\ref{ass2} holds, $\epsilon \in (0, \frac {1}{2 \overline d_{\max}})$, $\alpha \in (0,1-\epsilon \overline d_{\max})$, and $k_x, k_r \in (0, 1)$, then $\boldsymbol x^{p,*} \triangleq \lim_{k \to \infty} \mathbb E[\boldsymbol x^p(k)]$ exists and is given by
\begin{equation*}
\boldsymbol x^{p,*}=(\alpha \boldsymbol {\mathrm I}+\epsilon \overline L)^{-p}\alpha^p \boldsymbol r^*.
\end{equation*}
%\eqref{eqle1} is the expected steady-state value of the $p$th stage.
\end{lemma}
\vspace*{5pt}

\begin{proof}
Without loss of generality, only the first-stage state is analyzed here. The proof for the other stages is similar and is hence omitted. 

Replacing $k-\tau$ with $k$ in \eqref{eq4-2} yields
\begin{equation}\label{eq5-2}
\begin{split}
\hat x_i^1(k+1|k)&=\hat x_i^1(k|k-1)+k_x(x_i^1(k)\\
                 &\quad -\hat x_i^1(k|k-1))+u_i^1(k).
\end{split}
\end{equation}
Let $e_i^1(k)=x_i^1(k)-\hat x_i^1(k|k-1)$. Subtracting \eqref{eq5-2} from \eqref{eq4-5} leads to
\begin{equation}\label{eq5-3}
e_i^1(k+1)=e_i^1(k)-k_xe_i^1(k)=(1-k_x)e_i^1(k).
\end{equation}
Applying \eqref{eq4-3} recursively gives
\begin{equation}\label{eq5-4}
\begin{split}
 &\hat x_i^1(k|k-\tau)\\
=\,&\hat x_i^1(k-1|k-\tau)+u_i^1(k-1)\\
=\,&\hat x_i^1(k-\tau+1|k-\tau)+\sum_{l=2}^{\tau}u_i^1(k-\tau+l-1)\\
=\,&\hat x_i^1(k-\tau|k-\tau-1)+\sum_{l=2}^{\tau}u_i^1(k-\tau+k-1)\\
 &+u_i^1(k-\tau)+k_x(x_i^1(k-\tau)\\
 &-\hat x_i^1(k-\tau|k-\tau-1))\\
=\,&\hat x_i^1(k-\tau|k-\tau-1)+\sum_{l=1}^{\tau}u_i^1(k-\tau+l-1)\\
 &+k_x(x_i^1(k-\tau)-\hat x_i^1(k-\tau|k-\tau-1)).
\end{split}
\end{equation}
Similarly, applying \eqref{eq4-5} recursively yields
\begin{equation}\label{eq5-5}
\begin{split}
x_i^1(k)=x_i^1(k-\tau)+\sum_{l=1}^{\tau}u_i^1(k-\tau+l-1).
\end{split}
\end{equation}
Subtracting \eqref{eq5-5} from \eqref{eq5-4} leads to
\begin{equation}\label{eq5-6}
\begin{split}
 &\hat x_i^1(k|k-\tau)\\
=\,&x_i^1(k)+k_x(x_i^1(k-\tau)-\hat x_i^1(k-\tau|k-\tau-1))\\
 &+\hat x_i^1(k-\tau|k-\tau-1)-x_i^1(k-\tau)\\
=\,&x_i^1(k)+k_xe_i^1(k-\tau)-e_i^1(k-\tau)\\
=\,&x_i^1(k)+(k_x-1)e_i^1(k-\tau).
\end{split}
\end{equation}
It follows from \eqref{eq5-3} and \eqref{eq5-6} that
\begin{equation}\label{eq5-7}
\hat x_i^1(k|k-\tau)=x_i^1(k)-e_i^1(k-\tau+1).
\end{equation}
For the reference signals, define $e'_i(k)=\hat r_i(k)-\hat r_i(k|k-1)$. It then follows that
\begin{equation}\label{eq5-8}
e'_i(k+1)=e'_i(k)-k_r e'_i(k)=(1-k_r)e'_i(k),
\end{equation}
and
\begin{equation}\label{eq5-9}
\hat r_i(k|k-\tau)=\hat r_i(k)-e'_i(k-\tau+1).
\end{equation}

Using \eqref{eq5-7} and \eqref{eq5-9}, it follows from \eqref{eq4-5} that
\begin{equation}\label{eq5-10}
\begin{split}
 &x_i^1(k+1)\\
%=&x_i^1(K)+u_i^1(K)\\
=\, &x_i^1(k)-\epsilon \sum_{j=1}^N \theta_{ij} a_{ij} \left (x_i^1(k)-x_j^1(k) \right )\\
 &+\epsilon \sum_{j=1}^N \theta_{ij} a_{ij} \left(e_i^1(k-\tau+1)-e_j^1(k-\tau+1)\right)\\
 &+\alpha \left(\hat r_i(k)-e'_i(k-\tau+1) \right )\\
 &-\alpha \left (x_i^1(k)-e_i^1(k-\tau+1) \right ).
\end{split}
\end{equation}
Since $x_i^1(k)$ and $e_i^1(k-\tau+1)$ are independent of $\theta_{ij}$ at time $k$, taking mathematical expectation on both sides of \eqref{eq5-10} yields
\begin{equation}\label{eq5-11}
\begin{split}
 &\mathbb E[x_i^1(k+1)]\\
=\, &\mathbb E[x_i^1(k)]-\epsilon \sum_{j=1}^N \mathbb E[\theta_{ij}] a_{ij} \mathbb E \left [x_i^1(k)-x_j^1(k) \right ] \\
 &+\epsilon \sum_{j=1}^N \mathbb E[\theta_{ij}] a_{ij} \mathbb E \left [e_i^1(k-\tau+1)-e_j^1(k-\tau+1) \right]\\
 &+\alpha \left( \mathbb E[\hat r_i(k)]-\mathbb E[e'_i(k-\tau+1)] \right )\\
 &-\alpha \left (\mathbb E[x_i^1(k)]-\mathbb E[e_i^1(k-\tau+1)] \right ).
\end{split}
\end{equation}
It then follows from \eqref{eq3-2} that $\mathbb E[\theta_{ij}]=1-p_{ij}$, which together with \eqref{eq5-11} leads to
\begin{equation*}
\begin{split}
 &\mathbb E[x_i^1(k+1)]\\
=\, &\mathbb E[x_i^1(k)]-\epsilon \sum_{j=1}^N (1-p_{ij}) a_{ij} \left (\mathbb E [x_i^1(k)]-\mathbb E[x_j^1(k)] \right ) \\
 &+\epsilon \sum_{j=1}^N (1-p_{ij}) a_{ij} \left (\mathbb E [e_i^1(k-\tau+1)]-\mathbb E[e_j^1(k-\tau+1)] \right)\\
 &+\alpha \left( \mathbb E[\hat r_i(k)]-\mathbb E[e'_i(k-\tau+1)] \right )\\
 &-\alpha \left (\mathbb E[x_i^1(k)]-\mathbb E[e_i^1(k-\tau+1)] \right ).
\end{split}
\end{equation*}
Let
\begin{equation*}
\begin{split}
           \mathbb E[\Delta x_i^1(k)] &\triangleq \mathbb E[x_i^1(k)]-\mathbb E[x_i^1(k-1)]\\
           \mathbb E[\Delta e_i^1(k)] &\triangleq \mathbb E[e_i^1(k)]-\mathbb E[e_i^1(k-1)]\\
        \mathbb E[\Delta \hat r_i(k)] &\triangleq \mathbb E[\hat r_i(k)]-\mathbb E[\hat r_i(k-1)]\\
  \mathbb E[\Delta e'_i(k)] &\triangleq \mathbb E[e'_i(k)]-\mathbb E[e'_i(k-1)].\\
\end{split}
\end{equation*}
It follows that
\begin{equation}\label{eq5-12}
\begin{split}
 &\mathbb E[\Delta x_i^1(k+1)]\\
=\,&\mathbb E[\Delta x_i^1(k)]-\epsilon \sum_{j=1}^{N} (1-p_{ij}) a_{ij}(\mathbb E[\Delta x_i^1(k)]-\mathbb E[\Delta  x_j^1(k)])\\
 &-\epsilon \sum_{j=1}^{N} (1-p_{ij}) a_{ij}(\Delta \mathbb E[e_j^1(k-\tau+1)]\\
 &-\mathbb E[\Delta e_i^1(k-\tau+1)])\\
 &+\alpha \mathbb E[\Delta \hat r_i(k)]-\alpha \mathbb E[\Delta e'_i(k-\tau+1)]\\
 &-\alpha \mathbb E[\Delta x_i^1(k)]+\alpha \mathbb E[\Delta e_i^1(k-\tau+1)].
\end{split}
\end{equation}

% Recall $\mathbb E[\widetilde L] =[\mathbb E[[l_{ij}]_m]]$ as the expected value of the Laplacian matrix, i.e.,

Define
\begin{equation*}
\begin{split}
       \mathbb E[\Delta \boldsymbol x^1(k)]=[\mathbb E[\Delta x_1^1(k)], \mathbb E[\Delta x_2^1(k)], \dots, \mathbb E[\Delta x_N^1(k)]]^\mathrm T\\
\mathbb E[\Delta {\boldsymbol e}^1(k-\tau+1)]=[\mathbb E[\Delta e_1^1(k-\tau+1)], \mathbb E[\Delta e_2^1(k-\tau+1)], \\
                               \qquad  \dots, \mathbb E[\Delta e_N^1(k-\tau+1)]]^\mathrm T \\
    \mathbb E[\Delta \hat {\boldsymbol r}(k)]=[\mathbb E[\Delta \hat r_1(k)], \mathbb E[\Delta \hat r_2(k)], 
                                \dots, \mathbb E[\Delta \hat r_N(k)]]^\mathrm T\\
\mathbb E[\Delta \boldsymbol {e'}(k-\tau+1)]=[\mathbb E[\Delta e'_1(k-\tau+1)], 
                                \mathbb E[\Delta e'_2(k-\tau+1)], \\
                                 \qquad \dots, \mathbb E[\Delta e'_N(k-\tau+1)]]^\mathrm T.
\end{split}
\end{equation*}
Eq.~\eqref{eq5-12} can be written as
\begin{equation}\label{eq5-13}
\begin{split}
\mathbb E[\Delta \boldsymbol x^1(k+1)]=&\,[(1-\alpha)\boldsymbol {\mathrm I}-\epsilon \overline L] \mathbb E[\Delta \boldsymbol x^1(k)]\\
                &+(\alpha \boldsymbol {\mathrm I} +\epsilon \overline L)\mathbb E[\Delta \boldsymbol {e}^1(k-\tau+1)]\\
                &+\alpha \mathbb E[\Delta \hat {\boldsymbol r}(k)]-\alpha \mathbb E[\Delta \boldsymbol {e'}(k-\tau+1)].
\end{split}
\end{equation}
Using \eqref{eq5-3} and \eqref{eq5-8}, Eq.~\eqref{eq5-13} can be rewritten in a compact form as
\begin{equation}\label{eq5-14}
\begin{split}
 &\left [
\begin{array}{ccc}
\mathbb E[\Delta \boldsymbol x^1(k+1)]\\
\mathbb E[\Delta \boldsymbol {e}^1(k-\tau+2)]\\
\mathbb E[\Delta \boldsymbol {e'}(k-\tau+2)]
\end{array}
\right ]\\
=&M \left [
\begin{array}{ccc}
\mathbb E[\Delta \boldsymbol x^1(k)]\\
\mathbb E[\Delta \boldsymbol {e}^1(k-\tau+1)]\\
\mathbb E[\Delta \boldsymbol {e'}(k-\tau+1)]
\end{array}
\right ]
+\alpha
\left [
\begin{array}{ccc}
\mathbb E[\Delta \hat {\boldsymbol r}(k)]\\
0\\
0
\end{array}
\right ],
\end{split}
\end{equation}
where
\begin{equation}\label{eq5-15}
\begin{split}
M \triangleq
\left [
\begin{array}{ccc}
(1-\alpha)\boldsymbol {\mathrm I}-\epsilon \overline L& \alpha \boldsymbol {\mathrm I}+\epsilon \overline L&-\alpha\\
0&(1-k_x)\boldsymbol {\mathrm I}&0\\
0&0&(1-k_r)\boldsymbol {\mathrm I}
\end{array}
\right ].
\end{split}
\end{equation}
%Let
%\begin{equation*}
%       \boldsymbol r(k) \triangleq [r_1(k), r_2(k), \dots, r_N(k)]^\mathrm T.
%\end{equation*}
Taking mathematical expectation on both sides of \eqref{eq3-4} leads to
\begin{equation*}
\mathbb E[r_i(k+1)] = \mathbb E[r_i(k)]+\mathbb E[v_i(k)].
\end{equation*}
Based on Assumption \ref{ass2} and \eqref{eq5-1}, one has
\begin{equation}\label{eq5-16}
\lim_{k \to \infty} \mathbb E[\hat r_i(k)] \to r_i^*(k),
\end{equation}
which further leads to $\lim_{k \to \infty} \mathbb E[\Delta \hat {\boldsymbol r}(k)] \to 0$.

Due to \eqref{eq5-16}, as $k \to \infty$, Eq.~\eqref{eq5-14} becomes
\begin{equation}\label{eq5-17}
\begin{split}
 &\left [
\begin{array}{ccc}
\mathbb E[\Delta \boldsymbol x^1(k+1)]\\
\mathbb E[\Delta \boldsymbol {e}^1(k-\tau+2)]\\
\mathbb E[\Delta \boldsymbol {e'}(k-\tau+2)]
\end{array}
\right ]
=
M
\left [
\begin{array}{ccc}
\mathbb E[\Delta \boldsymbol x^1(k)]\\
\mathbb E[\Delta \boldsymbol {e}^1(k-\tau+1)]\\
\mathbb E[\Delta \boldsymbol {e'}(k-\tau+1  )]
\end{array}
\right ].
\end{split}
\end{equation}
It follows that the system~\eqref{eq5-17} is asymptotically stable if and only if all the eigenvalues of the matrix $M$ lie within the unit circle. Furthermore, \eqref{eq5-15} suggests that the eigenvalues of the matrix $M$ coincide with the  eigenvalues of the matrices $(1-\alpha)\boldsymbol {\mathrm I}-\epsilon \overline L$, $(1-k_x)\boldsymbol {\mathrm I}$ and $(1-k_r)\boldsymbol {\mathrm I}$.

For the matrices $(1-k_x)\boldsymbol {\mathrm I}$ and $(1-k_r)\boldsymbol {\mathrm I}$, by choosing $k_x, k_r \in (0, 1)$,  all the eigenvalues of the two matrices are within the unit circle. For the matrix $(1-\alpha)\boldsymbol {\mathrm I}-\epsilon \overline L$, it follows from \eqref{eq3-3} that $\overline L= \sum_{m=1}^s \mathbb P(\widetilde{\mathcal G}=\mathcal G_m) L_m$. Denote the $i$th eigenvalue of $\overline L$ by $\lambda_{\overline L, i}$. According to the Gerschgorin disc theorem, the expected Laplacian matrix has all its eigenvalues located within $[0,2 \overline d_{\max}]$, where $\overline d_{\max}$ is the maximum degree of the expected graph $\overline {\mathcal G}$. It thus follows that
\begin{equation*}
0=\lambda_{\overline L, 1} < \lambda_{\overline L, 2} \le \dots \le \lambda_{\overline L, N} \le 2 \overline d_{\max}.
\end{equation*}
Since $\overline L$ is a symmetric matrix, the matrix $(1-\alpha) \boldsymbol {\mathrm I}_N-\epsilon \overline L$ is  symmetric. Consequently, all the eigenvalues are real and are given by
\begin{equation}\label{eq5-18}
\lambda_i=1-\alpha-\epsilon \lambda_{\overline L, i}.
\end{equation}
Since $\alpha \in (0,1-\epsilon \overline d_{\max})$ and $\epsilon \in (0, \frac{1}{2 \overline d_{\max}})$, it follows that the eigenvalues of \eqref{eq5-18} satisfy $\lambda_i \in (0,1), ~i \in \{1, \dots, N \}$.
As all eigenvalues of $M$ are within the unit circle, the dynamical system~\eqref{eq5-17} is asymptotically stable.

Let $\boldsymbol x^{1, *}$ denote the expected steady state. That is,
\begin{equation*}
\boldsymbol x^{1,*} \triangleq \lim_{k \to \infty} \mathbb E[\boldsymbol x^1(k)]= [x_1^{1,*}, x_2^{1,*}, \dots, x_N^{1,*}]^T, %\\
%x^{1,*} \in \mathbb R^N,
\end{equation*}
where $x_i^{1, *} \triangleq \lim_{K \to \infty} \mathbb E[x_i^1(k)], i=1, 2, \dots, N$, represents the expected steady state at stage $1$ of agent $i$. As the asymptotic convergence of \eqref{eq5-17} is ensured, $\boldsymbol x^{1, *}$ is well defined.
The system~\eqref{eq5-17} is asymptotically stable, i.e., $\mathbb E[\Delta \boldsymbol {e}^1(k-\tau+2)] \to 0$, $\mathbb E[\Delta \boldsymbol {e'}(k-\tau+2)] \to 0$, $\mathbb E[\Delta \boldsymbol x^1(k)] \to 0$ as $k \to \infty$. As a result, 
\begin{equation}\label{eq5-19}
\lim_{k \to \infty} \mathbb E[\boldsymbol x^1(k+1)] \to \mathbb E[\boldsymbol x^1(k)] \to \boldsymbol x^{1, *}.
\end{equation}
It follows from \eqref{eq5-6} and \eqref{eq5-8} that $\mathbb E[e_i^1(k)] \to 0$ and $\mathbb E[e'_i(k)] \to 0$ as $k \to \infty$, which further leads \eqref{eq5-7} and \eqref{eq5-9} to
\begin{equation}\label{eq5-20}
\begin{split}
\mathbb E[\hat x_i^1(k|k-\tau)] &\to \mathbb E[x_i^1(k)]\\
  \mathbb E[\hat r_i(k|k-\tau)] &\to \mathbb E[\hat r_i(k)], ~\text{as $k \to \infty$}.
\end{split}
\end{equation}
Eq.~\eqref{eq4-6} gives
\begin{equation}\label{eq5-21}
\begin{split}
\mathbb E[\boldsymbol x^1(k+1)]=&\mathbb E[\boldsymbol x^1(k)]-\epsilon  \overline L \mathbb E[\hat {\boldsymbol x}^1(k|k-\tau)]\\
                    &+ \alpha \left (\mathbb E[\hat {\boldsymbol r}(k)]-\mathbb E[\hat {\boldsymbol x}^1(k|k-\tau)] \right ).
\end{split}
\end{equation}
Using \eqref{eq5-16}, \eqref{eq5-19} and \eqref{eq5-20}, it follows from \eqref{eq5-21} that
\begin{equation*}
\begin{split}
\boldsymbol x^{1,*}=\boldsymbol x^{1,*}-\epsilon  \overline L \boldsymbol x^{1,*}+ \alpha \left (\boldsymbol r^*-\boldsymbol x^{1,*} \right ).
\end{split}
\end{equation*}
The expected steady-state equilibrium at the first stage is then given by
\begin{equation}\label{eq5-22}
\boldsymbol x^{1,*}=(\alpha \boldsymbol {\mathrm I}+\epsilon \overline L)^{-1} \alpha \boldsymbol r^*.
\end{equation}
For the remaining $n-1$ stages, it can be shown similarly that
\begin{equation}\label{eq5-23}
\boldsymbol x^{p,*}=(\alpha \boldsymbol {\mathrm I}+\epsilon \overline L)^{-1} \alpha \boldsymbol x^{p-1,*},
\end{equation}
where $\boldsymbol x^{p,*} \triangleq [x_1^{p,*}, x_2^{p,*}, \dots, x_N^{p,*}]^\mathrm T \in \mathbb R^N, p=1,2, \dots, n$, represents the expected steady state at stage $p$ and $x_i^{p,*} \triangleq \lim_{k \to \infty} \mathbb E[x_i^p(k)], i=1, 2, \dots, N$. Combining \eqref{eq5-22} and \eqref{eq5-23} gives
\begin{equation*}
\boldsymbol x^{p,*}=(\alpha \boldsymbol {\mathrm I}+\epsilon \overline L)^{-p}\alpha^p \boldsymbol r^*.
\end{equation*}
The proof is thus completed.
%Summarizing the above results giving the following Lemma \ref{lem4},
\end{proof}

The main result of this section is the following theorem.
\begin{theorem}
\label{theor:conv}
For the system~\eqref{eq4-5}  with a connected communication network, if Assumption~\ref{ass2} holds, $\epsilon \in (0, \frac {1}{2 \overline d_{\max}})$, $\alpha \in (0, 1-\epsilon \overline d_{\max})$, and $k_x, k_r \in (0, 1)$, then
\begin{equation*}
\limsup_{k \to \infty} \left \| \bar r^* \boldsymbol 1_N-\boldsymbol x^{n,*} \right \|_2 \le \left ( \frac{\alpha}{\alpha+\epsilon \lambda_{\overline L, 2}} \right)^n \|\widetilde {\boldsymbol r}^*\|_2,
\end{equation*}
where $\widetilde {\boldsymbol r}^* \triangleq (\boldsymbol r^*-\bar r^* \boldsymbol 1_N)$ and $\bar r^* \triangleq \frac {1} {N} \sum_{i=1}^N r_i^*$. %then \eqref{eqth1} is the expected steady-state estimation error of the $n$th stage.
\end{theorem}
\begin{proof}
It follows from Lemma \ref{lem4} that
\begin{equation*}
\boldsymbol x^{p,*}=(\alpha \boldsymbol {\mathrm I}+\epsilon \overline L)^{-p}\alpha^p \boldsymbol r^*.
\end{equation*} 
The eigenvalues of the matrix $(\alpha \boldsymbol {\mathrm I}+\epsilon \overline L)^{-p}$ are given by
\begin{equation*}
\lambda_i'= \left (\frac {1}{\alpha+\epsilon \lambda_{\overline L, i}} \right)^p,
\end{equation*}
where $\lambda_{\overline L, i}$ denotes the $i$th eigenvalue of the matrix $\overline L$, and $\boldsymbol v_i$ is the corresponding eigenvector. The columns of the matrix $V=[\boldsymbol v_1, \boldsymbol v_2, \dots,\boldsymbol v_N]$ are orthonormal. Consequently, $\overline L$ can be expressed as $\overline L=V \Lambda_{\overline L}V^T=\sum_{i=1}^N \lambda_{\overline L,i} \boldsymbol v_i \boldsymbol v_i^T$, where $\Lambda_{\overline L}= \text{diag} \{\lambda_{\overline L, 1}, \lambda_{\overline L, 2}, \dots, \lambda_{\overline L, N}\}$ and $\boldsymbol v_1 = \frac {1}{\sqrt N} \boldsymbol 1_N$. Hence, 
\begin{equation}\label{eq5-24}
\begin{split}
\boldsymbol x^{p,*}=&(\alpha \boldsymbol {\mathrm I}+\epsilon \overline L)^{-p}\alpha^p \boldsymbol r^*\\
                   =&(\alpha V \boldsymbol {\mathrm I} V^T+\epsilon V \Lambda_{\overline L}V^T)^{-p} \alpha^p \boldsymbol r^*\\
                   =&(V(\alpha \boldsymbol {\mathrm I}+\epsilon \Lambda_{\overline L})V^T)^{-p}\alpha^p \boldsymbol r^*\\
                   =&\sum_{i=1}^N \left(\frac {1}{\alpha+\epsilon \lambda_{\overline L, i}} \right )^p \boldsymbol v_i \boldsymbol v_i^T \alpha^p \boldsymbol r^*\\
                   =&\sum_{i=1}^N \left(\frac {\alpha}{\alpha+\epsilon \lambda_{\overline L, i}} \right )^p \boldsymbol v_i \boldsymbol v_i^T \boldsymbol r^*\\
                   =&\sum_{i=2}^N \left(\frac {\alpha}{\alpha+\epsilon \lambda_{\overline L, i}} \right )^p \boldsymbol v_i \boldsymbol v_i^T \boldsymbol r^*+\boldsymbol v_1 \boldsymbol v_1^T \boldsymbol r^*\\
                   =&\sum_{i=2}^N \left(\frac {\alpha}{\alpha+\epsilon \lambda_{\overline L, i}} \right )^p \boldsymbol v_i \boldsymbol v_i^T \boldsymbol r^*+\frac {1}{N} \boldsymbol 1_{N \times N} \boldsymbol r^*\\
                   =&\sum_{i=2}^N \left(\frac {\alpha}{\alpha+\epsilon \lambda_{\overline L, i}} \right )^p \boldsymbol v_i \boldsymbol v_i^T \boldsymbol r^*+\bar r^* \boldsymbol 1_N.
\end{split}
\end{equation}
It follows from \eqref{eq5-24} that
\begin{equation*}
\begin{split}
& \quad \limsup_{k \to \infty} \|\bar {r}^* \boldsymbol 1_N-\boldsymbol x^{n,*} \|_2 \nonumber \\
&=\left \|\sum_{i=2}^N \left (\frac {\alpha}{\alpha+\epsilon \lambda_{\overline L, i}} \right )^n \boldsymbol v_i \boldsymbol v_i^T \boldsymbol r^* \right \|_2\\
                                          &\le \left (\frac {\alpha}{\alpha+\epsilon \lambda_{\overline L, 2}} \right )^n  \left \|\sum_{i=2}^N \boldsymbol v_i \boldsymbol v_i^T \boldsymbol r^* \right \|_2\\
                                          &\le \left (\frac {\alpha}{\alpha+\epsilon \lambda_{\overline L, 2}} \right )^n  \left \|\sum_{i=1}^N \boldsymbol v_i \boldsymbol v_i^T \boldsymbol r^* -\boldsymbol v_1 \boldsymbol v_1^T \boldsymbol r^* \right \|_2\\
                                          &\le \left (\frac {\alpha}{\beta} \right )^n \| \widetilde {\boldsymbol r}^* \|_2,
\end{split}
\end{equation*}
where $\widetilde {\boldsymbol r}^* :=(\boldsymbol r^*-\bar {r}^* \boldsymbol 1_N)$.
\end{proof}
On the one hand, Theorem~\ref{theor:conv} shows the possibility of achieving mean-squared DAT in the presence of input delay, reference noise, and packet-drops. As the stage number $n$ goes to infinity, the tracking error will approach zero, i.e, the control objective \eqref{eq3-5} will be achieved. On the other hand, a larger stage number $n$ will induce a higher communication/computation cost for the agents, indicating that there exists trade-off between the tracking error and communication/computation cost.

\section{Simulation}
\label{sec:sim}
In this section, a numerical example is presented to verify Theorem~\ref{theor:conv}. 

A system consisting of $N=4$ nodes is considered. The communication network topology is given in Fig.~\ref{fig:6-1}.
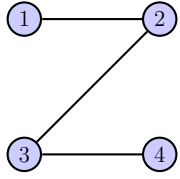
\begin{figure}[!htb]
\centering
\begin{tikzpicture}[
    -,>=stealth',shorten >=0.2pt,auto,node distance=3cm,thick,
    node/.style={scale=.6,circle,fill=blue!20,draw,font=\sffamily\Large\mdseries},
    function/.style={scale=2,rectangle,draw,fill=cyan!20,font=\sffamily\small\mdseries}]
    \node[node] (n1) {$1$};
    \node[node] [right of=n1](n2) {$2$} ;
    \node[node] [below of=n1](n3) {$3$} ;
    \node[node] [below of=n2](n4) {$4$} ;
    \draw[-] (n2)--(n1);
	\draw[-] (n3)--(n2);
	\draw[-] (n4)--(n3);
\end{tikzpicture}
\caption{The communication network topology}\label{fig:6-1}
\end{figure}
Assume that the packet-drop probability $p_{ij}=0.5$, $\forall \, (i,j) \in \mathcal{E}$. The actual network topology can vary from one of the eight network topologies shown in Fig.~\ref{fig:6-2} due to packet-drops. 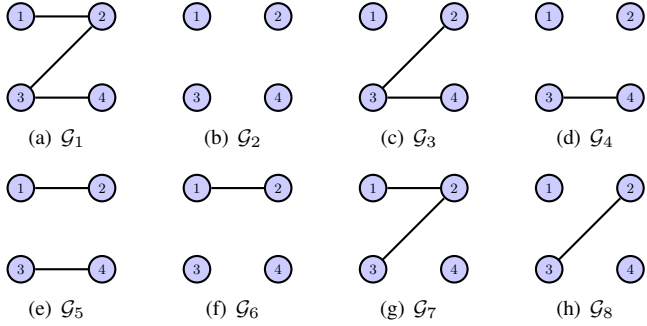
\begin{figure}[!htb]
\centering \subfigure[$\mathcal G_1$]{
\label{fig6-2:subfig:a}
\begin{tikzpicture}[
    -,>=stealth',shorten >=0.2pt,auto,node distance=1.8cm,thick,
    node/.style={scale=0.6,circle,fill=blue!20,draw,font=\sffamily\small\mdseries},
    function/.style={scale=2,rectangle,draw,fill=cyan!20,font=\sffamily\small\mdseries}]
    \node[node] (n1) {$1$};
    \node[node] [right of=n1](n2) {$2$} ;
    \node[node] [below of=n1](n3) {$3$} ;
    \node[node] [below of=n2](n4) {$4$} ;
    \draw[-] (n2)--(n1);
	\draw[-] (n3)--(n2);
	\draw[-] (n4)--(n3);
\end{tikzpicture}}
\hspace{0.25in}\subfigure[$\mathcal G_2$]{
\label{fig6-2:subfig:b}
\begin{tikzpicture}[
    -,>=stealth',shorten >=0.2pt,auto,node distance=1.8cm,thick,
    node/.style={scale=0.6,circle,fill=blue!20,draw,font=\sffamily\small\mdseries},
    function/.style={scale=2,rectangle,draw,fill=cyan!20,font=\sffamily\small\mdseries}]
    \node[node] (n1) {$1$};
    \node[node] [right of=n1](n2) {$2$} ;
    \node[node] [below of=n1](n3) {$3$} ;
    \node[node] [below of=n2](n4) {$4$} ;
%   \draw[-] (n2)--(n1);
%	\draw[-] (n3)--(n2);
%	\draw[-] (n4)--(n3);
\end{tikzpicture}}
\hspace{0.25in}\subfigure[$\mathcal G_3$]{
\label{fig6-2:subfig:c}
\begin{tikzpicture}[
    -,>=stealth',shorten >=0.2pt,auto,node distance=1.8cm,thick,
    node/.style={scale=0.6,circle,fill=blue!20,draw,font=\sffamily\small\mdseries},
    function/.style={scale=2,rectangle,draw,fill=cyan!20,font=\sffamily\small\mdseries}]
    \node[node] (n1) {$1$};
    \node[node] [right of=n1](n2) {$2$} ;
    \node[node] [below of=n1](n3) {$3$} ;
    \node[node] [below of=n2](n4) {$4$} ;
%   \draw[-] (n2)--(n1);
	\draw[-] (n3)--(n2);
	\draw[-] (n4)--(n3);
\end{tikzpicture}}
\hspace{0.25in}\subfigure[$\mathcal G_4$]{
\label{fig6-2:subfig:d}
\begin{tikzpicture}[
    -,>=stealth',shorten >=0.2pt,auto,node distance=1.8cm,thick,
    node/.style={scale=0.6,circle,fill=blue!20,draw,font=\sffamily\small\mdseries},
    function/.style={scale=2,rectangle,draw,fill=cyan!20,font=\sffamily\small\mdseries}]
    \node[node] (n1) {$1$};
    \node[node] [right of=n1](n2) {$2$} ;
    \node[node] [below of=n1](n3) {$3$} ;
    \node[node] [below of=n2](n4) {$4$} ;
%   \draw[-] (n2)--(n1);
%	\draw[-] (n3)--(n2);
	\draw[-] (n4)--(n3);
\end{tikzpicture}}
\vfill
\subfigure[$\mathcal G_5$]{
\label{fig6-2:subfig:e}
\begin{tikzpicture}[
    -,>=stealth',shorten >=0.2pt,auto,node distance=1.8cm,thick,
    node/.style={scale=0.6,circle,fill=blue!20,draw,font=\sffamily\small\mdseries},
    function/.style={scale=2,rectangle,draw,fill=cyan!20,font=\sffamily\small\mdseries}]
    \node[node] (n1) {$1$};
    \node[node] [right of=n1](n2) {$2$} ;
    \node[node] [below of=n1](n3) {$3$} ;
    \node[node] [below of=n2](n4) {$4$} ;
    \draw[-] (n2)--(n1);
%	\draw[-] (n3)--(n2);
	\draw[-] (n4)--(n3);
\end{tikzpicture}}
\hspace{0.25in}\subfigure[$\mathcal G_6$]{
\label{fig6-2:subfig:f}
\begin{tikzpicture}[
    -,>=stealth',shorten >=0.2pt,auto,node distance=1.8cm,thick,
    node/.style={scale=0.6,circle,fill=blue!20,draw,font=\sffamily\small\mdseries},
    function/.style={scale=2,rectangle,draw,fill=cyan!20,font=\sffamily\small\mdseries}]
    \node[node] (n1) {$1$};
    \node[node] [right of=n1](n2) {$2$} ;
    \node[node] [below of=n1](n3) {$3$} ;
    \node[node] [below of=n2](n4) {$4$} ;
    \draw[-] (n2)--(n1);
%	\draw[-] (n3)--(n2);
%	\draw[-] (n4)--(n3);
\end{tikzpicture}}
\hspace{0.25in}\subfigure[$\mathcal G_7$]{
\label{fig6-2:subfig:g}
\begin{tikzpicture}[
    -,>=stealth',shorten >=0.2pt,auto,node distance=1.8cm,thick,
    node/.style={scale=0.6,circle,fill=blue!20,draw,font=\sffamily\small\mdseries},
    function/.style={scale=2,rectangle,draw,fill=cyan!20,font=\sffamily\small\mdseries}]
    \node[node] (n1) {$1$};
    \node[node] [right of=n1](n2) {$2$} ;
    \node[node] [below of=n1](n3) {$3$} ;
    \node[node] [below of=n2](n4) {$4$} ;
    \draw[-] (n2)--(n1);
	\draw[-] (n3)--(n2);
%	\draw[-] (n4)--(n3);
\end{tikzpicture}}
\hspace{0.25in}\subfigure[$\mathcal G_8$]{
\label{fig6-2:subfig:h}
\begin{tikzpicture}[
    -,>=stealth',shorten >=0.2pt,auto,node distance=1.8cm,thick,
    node/.style={scale=0.6,circle,fill=blue!20,draw,font=\sffamily\small\mdseries},
    function/.style={scale=2,rectangle,draw,fill=cyan!20,font=\sffamily\small\mdseries}]
    \node[node] (n1) {$1$};
    \node[node] [right of=n1](n2) {$2$} ;
    \node[node] [below of=n1](n3) {$3$} ;
    \node[node] [below of=n2](n4) {$4$} ;
%    \draw[-] (n2)--(n1);
	\draw[-] (n3)--(n2);
%	\draw[-] (n4)--(n3);
\end{tikzpicture}}
\caption{The possible actual network topologies}\label{fig:6-2}
\end{figure}
Their corresponding Laplacian matrices are given respectively by
\begin{align}
	L_1 &=\left [
\begin{array}{cccc}
1&-1&0&0\\
-1&2&-1&0\\
0&-1&2&-1\\
0&0&-1&1
\end{array}
\right ], & L_2&=\left [
\begin{array}{cccc}
0&0&0&0\\
0&0&0&0\\
0&0&0&0\\
0&0&0&0
\end{array}
\right ] \nonumber \\
L_3&=\left [
\begin{array}{cccc}
0&0&0&0\\
0&1&-1&0\\
0&-1&2&-1\\
0&0&-1&1
\end{array}
\right ], &
L_4&=\left [
\begin{array}{cccc}
0&0&0&0\\
0&0&0&0\\
0&0&1&-1\\
0&0&-1&1
\end{array}
\right ] \nonumber \\
L_5&=\left [
\begin{array}{cccc}
1&-1&0&0\\
-1&1&0&0\\
0&0&1&-1\\
0&0&-1&1
\end{array}
\right ], &
L_6&=\left [
\begin{array}{cccc}
1&-1&0&0\\
-1&1&0&0\\
0&0&0&0\\
0&0&0&0
\end{array}
\right ]\nonumber \\
L_7&=\left [
\begin{array}{cccc}
1&-1&0&0\\
-1&2&-1&0\\
0&-1&1&0\\
0&0&0&0
\end{array}
\right ], &
L_8&=\left [
\begin{array}{cccc}
0&0&0&0\\
0&1&-1&0\\
0&-1&1&0\\
0&0&0&0
\end{array}
\right ]. \nonumber
\end{align}
% \setlength{\arraycolsep}{4pt}
% \begin{align}
% \begin{split}
% L_1&=\left [
% \begin{array}{cccc}
% 1&-1&0&0\\
% -1&2&-1&0\\
% 0&-1&2&-1\\
% 0&0&-1&1
% \end{array}
% \right ], &
% L_2&=\left [
% \begin{array}{cccc}
% 0&0&0&0\\
% 0&0&0&0\\
% 0&0&0&0\\
% 0&0&0&0
% \end{array}
% \right ],\\
% L_3&=\left [
% \begin{array}{cccc}
% 0&0&0&0\\
% 0&1&-1&0\\
% 0&-1&2&-1\\
% 0&0&-1&1
% \end{array}
% \right ], &
% L_4&=\left [
% \begin{array}{cccc}
% 0&0&0&0\\
% 0&0&0&0\\
% 0&0&1&-1\\
% 0&0&-1&1
% \end{array}
% \right ],\\
% L_5&=\left [
% \begin{array}{cccc}
% 1&-1&0&0\\
% -1&1&0&0\\
% 0&0&1&-1\\
% 0&0&-1&1
% \end{array}
% \right ], &
% L_6&=\left [
% \begin{array}{cccc}
% 1&-1&0&0\\
% -1&1&0&0\\
% 0&0&0&0\\
% 0&0&0&0
% \end{array}
% \right ],\\
% L_7&=\left [
% \begin{array}{cccc}
% 1&-1&0&0\\
% -1&2&-1&0\\
% 0&-1&1&0\\
% 0&0&0&0
% \end{array}
% \right ], &
% L_8&=\left [
% \begin{array}{cccc}
% 0&0&0&0\\
% 0&1&-1&0\\
% 0&-1&1&0\\
% 0&0&0&0
% \end{array}
% \right ].
% \end{split}
% \end{align}
The expected Laplacian matrix $\overline L$ is given by
\begin{equation*}
\overline L=\mathbb P(\widetilde{\mathcal G}=\mathcal G_1)L_1+\mathbb P(\widetilde{\mathcal G}=\mathcal G_2)L_2+ \dots+ \mathbb P(\widetilde{\mathcal G}=\mathcal G_8)L_8.
\end{equation*}
Fig.~\ref{fig:6-3} shows the expected network topology $\overline {\mathcal G}$, the network topology corresponding to $\overline L$.
\begin{figure}[!htb]
\centering
\begin{tikzpicture}[
    -,>=stealth',shorten >=0.2pt,auto,node distance=3cm,thick,
    node/.style={scale=.6,circle,fill=blue!20,draw,font=\sffamily\Large\mdseries},
    function/.style={scale=2,rectangle,draw,fill=cyan!20,font=\sffamily\small\mdseries}]
    \node[node] (n1) {$1$};
    \node[node] [right of=n1](n2) {$2$} ;
    \node[node] [below of=n1](n3) {$3$} ;
    \node[node] [below of=n2](n4) {$4$} ;
    \draw[-] (n2)--(n1) node [above, midway] {0.5};
	\draw[-] (n3)--(n2) node [above, midway] {0.5};
	\draw[-] (n4)--(n3) node [above, midway] {0.5};
\end{tikzpicture}
\caption{The expected network topology}\label{fig:6-3}
\end{figure}
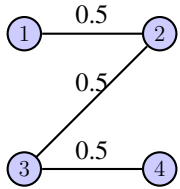
The reference signals are governed by \eqref{eq3-4}, with $\phi_i=0.01$ and $\psi_i=1$. Finally, the parameters are chosen as $\epsilon=\frac{1}{8}$, $\alpha=\frac{1}{2}$, and $k_x=k_r=\frac{1}{2}$. The input delay is set to $\tau=5$.

\begin{figure}[!htb]
\centering
%\subfigure[n=1.]{
%\label{fig:6-4:subfig:a} %% µÚÒ»·ùͼµÄ±êÇ©
%\includegraphics[scale=0.5]{6-41.eps}}
%\hspace{0.1in}
\subfigure[n=10.]{
\label{fig:6-4:subfig:b} %% µÚ¶þ·ùͼµÄ±êÇ©
\includegraphics[scale=0.5]{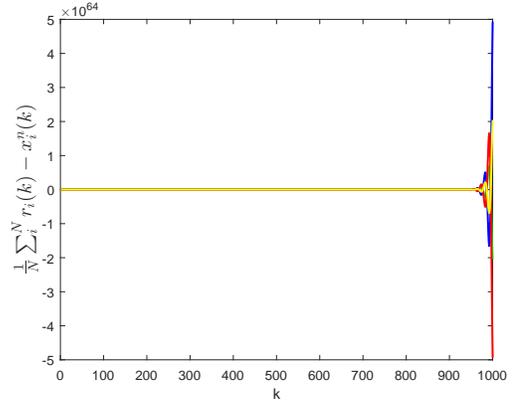}}
\hspace{1in} \subfigure[n=100.]{
\label{fig:6-4:subfig:c} %% µÚÈý·ùͼµÄ±êÇ©
\includegraphics[scale=0.5]{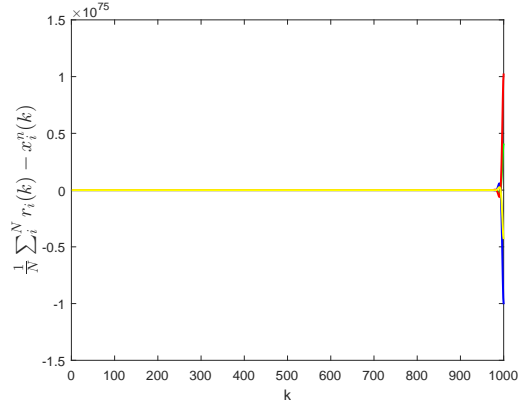}}
\caption{Tracking error of the $n$th stage of $x_i(k)$ under the DAT algorithm \eqref{eq4-1}.}
\label{fig:6-4:subfig}
\end{figure}

Fig.~\ref{fig:6-4:subfig} shows the tracking error $\frac{1}{N}\sum_{i=1}^N r_i(k)-x_i^n(k)$ under the proposed algorithm embedded in system \eqref{eq4-5} without employing the Kalman filter and state predictor to handle the input delays, packet-drops, and noisy reference signals. It can be observed that the system cannot track the average of the reference signals; instead, the tracking error diverges eventually.
\begin{figure}[!htb]
\centering
%\subfigure[n=1.]{
%\label{fig6-5:subfig:a} %% µÚÒ»·ùͼµÄ±êÇ©
%\includegraphics[scale=0.5]{6-51.eps}}
\hspace{1in} \subfigure[n=10.]{
\label{fig6-5:subfig:b} %% µÚ¶þ·ùͼµÄ±êÇ©
\includegraphics[scale=0.5]{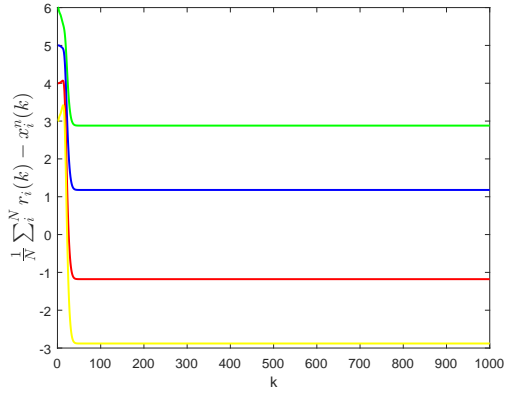}}
\hspace{1in} \subfigure[n=100.]{
\label{fig6-5:subfig:c} %% µÚÈý·ùͼµÄ±êÇ©
\includegraphics[scale=0.5]{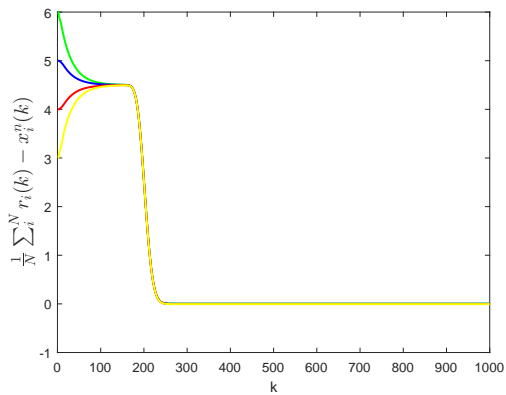}}
\caption{Tracking error of the $n$th stage of $x_i(k)$ under the DAT algorithm \eqref{eq4-1}.}
\label{fig6-5:subfig}
\end{figure}

Fig.~\ref{fig6-5:subfig} shows the tracking error of the proposed DAT algorithm with the Kalman filter and state predictor for different stage number $n$. It can be seen that in both cases the tracking error will finally reach a steady value. Furthermore, as the stage number gets larger, the tracking error gets smaller, which is consistent with Theorem~\ref{theor:conv}.
\begin{figure}[!htb]
\centering
%\subfigure[n=1.]{
%\label{fig6-6:subfig:a} %% µÚÒ»·ùͼµÄ±êÇ©
%\includegraphics[scale=0.5]{6-61.eps}}
\hspace{1in} \subfigure[n=10.]{
\label{fig6-6:subfig:b} %% µÚ¶þ·ùͼµÄ±êÇ©
\includegraphics[scale=0.5]{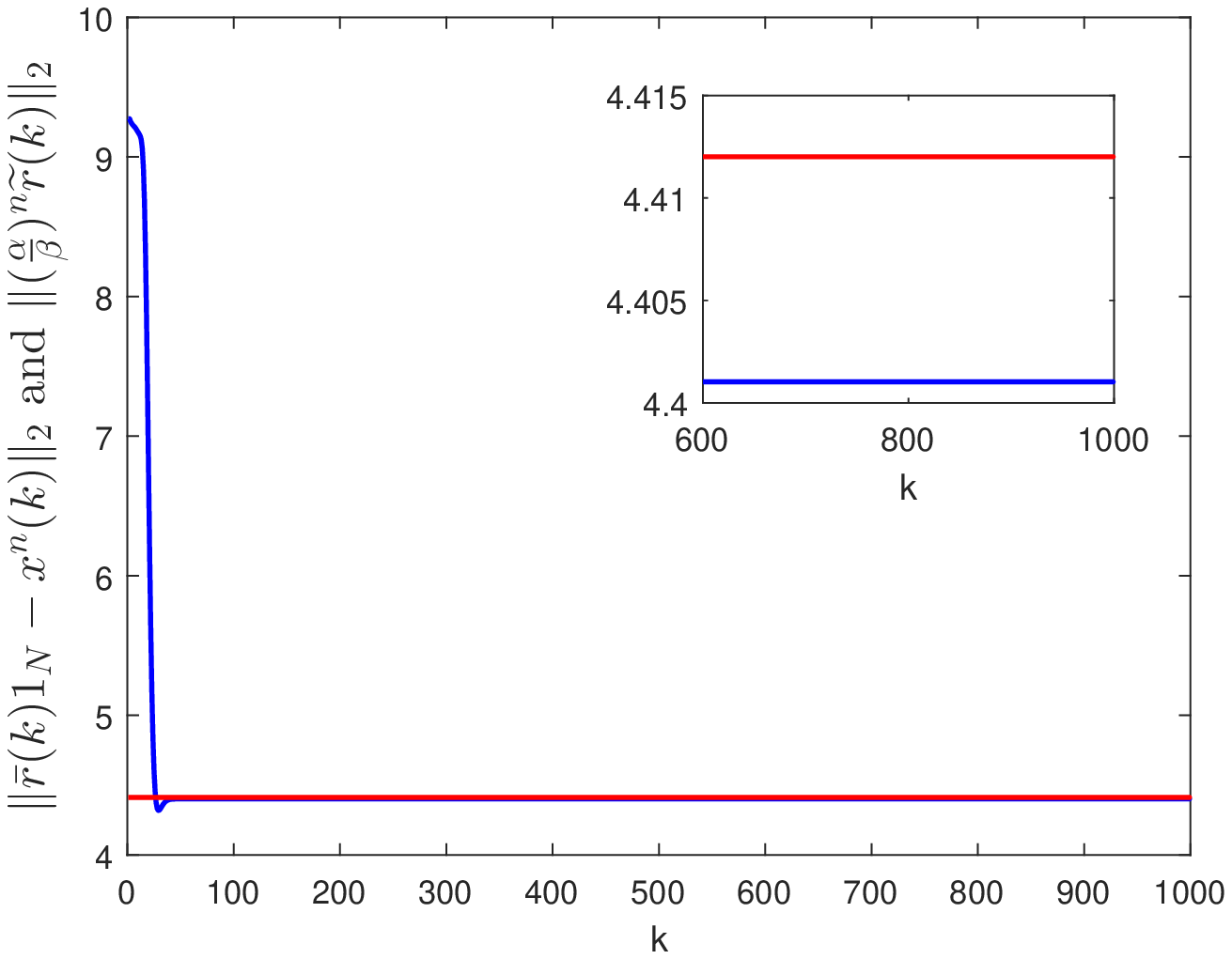}}
\hspace{1in} \subfigure[n=100.]{
\label{fig6-6:subfig:c} %% µÚÈý·ùͼµÄ±êÇ©
\includegraphics[scale=0.5]{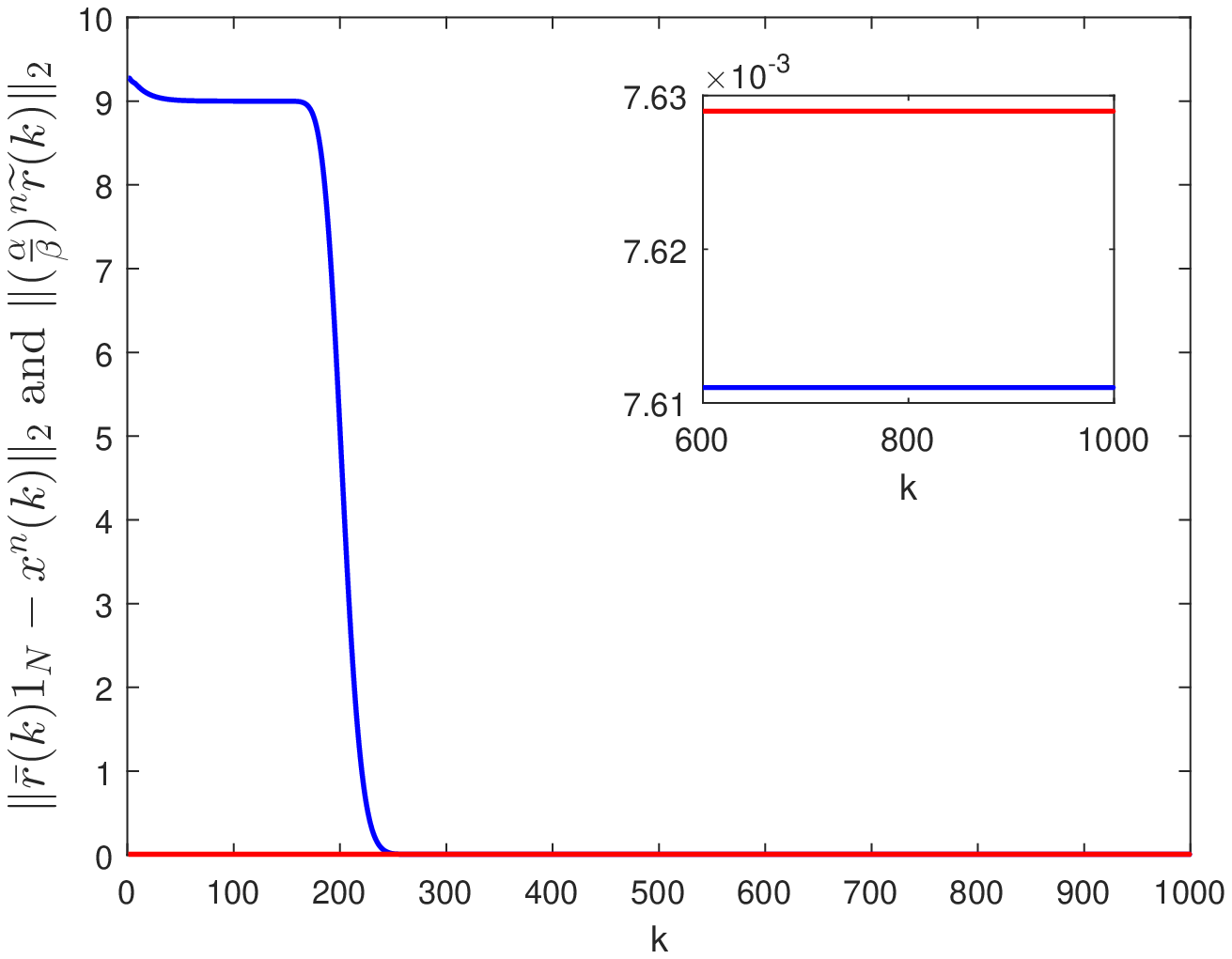}}
\caption{The trajectories of $\|\bar r(k) \boldsymbol 1_N-\boldsymbol x^n(k) \|_2$ and $\delta=\left (\frac {\alpha}{\beta} \right )^n \| \widetilde {\boldsymbol r}(k)\|_2$.}
\label{fig6-6:subfig}
\end{figure}

Fig.~\ref{fig6-6:subfig} shows respectively the trajectories of $\|\bar {r}(k) \boldsymbol 1_N-\boldsymbol x^n(k) \|_2$ and $\delta=\left (\frac {\alpha}{\beta} \right )^n \| \widetilde {\boldsymbol r}(k) \|_2$. The red line corresponds to the trajectory of $\left (\frac {\alpha}{\beta} \right )^n \| \widetilde {\boldsymbol r}(k) \|_2$, while the blue line to the trajectory of $\|\bar {r}(k) \boldsymbol 1_N-\boldsymbol x^n(k) \|_2$. It can be seen that as $k \to \infty$, the inequality $\limsup \|\bar {r}(k) \boldsymbol 1_N-\boldsymbol x^n(k) \|_2 \le \left (\frac {\alpha}{\beta} \right )^n \| \widetilde {\boldsymbol r}(k) \|_2$ holds eventually. %i.e., $\|\bar {r}^* \boldsymbol 1_N-x^{n,*} \|_2 \le \left (\frac {\alpha}{\beta} \right )^n \| \widetilde {\boldsymbol r}^* \|_2$. %It can be seen from Figs.~\ref{fig6-5:subfig} and \ref{fig6-6:subfig}, that the control objective \eqref{eq3-5} is achieved.

\section{Conclusion}
\label{sec:concl}

This paper demonstrates that DAT algorithms can be seriously hampered by reference noise, packet-drops, and input delays; however, it is still possible to achieve practical DAT by employing appropriate control techniques, such as Kalman filtering and predictive control, to deal with those negative effects. 

In summary, the following statements are drawn from this work.
\begin{itemize}
	\item An explicit expression of the expected stationary states of the agents is obtained and given in terms of the expected values of the references, which also depends on the control gains as well as the number of processing stages.
	\item The mean-squared tracking error is ultimately upper bounded by the average difference among the reference signals, and as the number of stages goes to infinity, the tracking error will varnish, achieving practical DAT in the sense of mean square.
\end{itemize}
These results shed new lights on the studies of distributed average tracking and cooperative control of multi-agent systems.

% In this paper, a distributed average tracking algorithm with input delays, random packet dropout is investigated, where the reference inputs are subject to random process and measurement noise. With the aid of Kalman filter, taking input delays and packet dropout into consider, we have developed a discrete-time DAT algorithm, which is capable of attenuating the effect of the noise. A predictive scheme has been proposed to compensate the effect of input delays.  We have proved the convergence of the proposed algorithm. Finally, we have verified  validate the derived result via a numerical example. Future work will be focused on DAT problem for multi-agent system with communication noise and delays among agents.

\bibliographystyle{IEEEtran}
\bibliography{IEEEabrv,Mybibfile2}

% \begin{IEEEbiography}{First A. Author}
% \end{IEEEbiography}

% \begin{IEEEbiography}{Second B. Author} 1.
% \end{IEEEbiography}

% \begin{IEEEbiography}{Third C. Author, Jr.}
% \end{IEEEbiography}
\end{document}